\documentclass[a4paper]{article}
\usepackage{a4wide}
\usepackage[utf8]{inputenc}
\usepackage[ngerman, english]{babel}
\usepackage{amsmath,amsthm,amssymb,amsfonts}
\usepackage{mathrsfs}
\usepackage{ifthen}
\usepackage{enumitem}
\usepackage{comment}
\usepackage{subcaption}
\usepackage{hyphsubst}
\usepackage{xcolor}
\usepackage{calc}
\usepackage{graphicx}
\usepackage{pgfplots}
\pgfplotsset{compat=1.5.1}
\usepackage{hyperref}
\usepackage[toc,page]{appendix}
\usepackage{cleveref}
\usepackage{leftindex}
\usepackage{algorithm}
\usepackage{algpseudocode}
\usepackage[bbgreekl]{mathbbol}
\DeclareSymbolFontAlphabet{\mathbbm}{bbold}
\DeclareSymbolFontAlphabet{\mathbb}{AMSb}
\usepackage{bbold}

\makeatletter
\def\@hspace#1{\begingroup\setlength\dimen@{#1}\hskip\dimen@\endgroup}
\makeatother

\newtheorem{theorem}{Theorem}[section]
\newtheorem{lemma}[theorem]{Lemma}

\newtheorem{corollary}[theorem]{Corollary}

\newcommand{\one}{\mathbb{1}}
\def\yobs{y^{\text{obs}}}

\def\BopIpsdtr{\mathcal{D}}

\def\rhofix{\rho^{\text{fix}}}

\newcommand{\dd}{\mathrm{d}}

\newcommand{\bpm}{\begin{pmatrix}}
\newcommand{\epm}{\end{pmatrix}}

\DeclareMathOperator{\rank}{rank}
\renewcommand{\dim}{\operatorname{dim}}

\DeclareMathOperator{\tr}{tr}
\DeclareMathOperator{\ran}{ran}

\DeclareMathOperator{\argmin}{argmin}

\DeclareMathOperator{\pois}{Pois}
\DeclareMathOperator{\Herm}{Herm}

\newcommand{\setC}{\mathbb{C}}

\newcommand{\setN}{\mathbb{N}}

\newcommand{\setR}{\mathbb{R}}

\newcommand{\calF}{\mathcal{F}}

\newcommand{\calI}{\mathcal{I}}

\newcommand{\calL}{\mathcal{L}}
\newcommand{\calM}{\mathcal{M}}

\newcommand{\calO}{\mathcal{O}}

\newcommand{\calS}{\mathcal{S}}

\definecolor{brickred}{rgb}{0.8, 0.25, 0.33}
\definecolor{bostonuniversityred}{rgb}{0.8, 0.0, 0.0}
\definecolor{cornellred}{rgb}{0.7, 0.11, 0.11}
\definecolor{corn}{rgb}{0.98, 0.93, 0.36}
\definecolor{schoolbusyellow}{rgb}{1.0, 0.85, 0.0}
\definecolor{TUblue}{rgb}{0,102,153}
\colorlet{TUbluelight}{TUblue!30!white}
\usepackage{fancyhdr}
\usepackage[bbgreekl]{mathbbol}
\usepackage{authblk}

\author[1]{Florian Oberender}
\affil[1]{Institut f\"ur Numerische und Angewandte Mathematik, Georg-August Universität Göttingen}
\title{On spurious fixed points in iterative maximum likelihood reconstruction for quantum tomography%
\footnote{The first author acknowledges support from DFG, CRC 1456 project 432680300.}}

\begin{document}

\maketitle
\begin{abstract}
\noindent
Maximum likelihood iteration is one of the most commonly used reconstruction algorithms in quantum tomography. The main appeal of the method is that it is easy to implement and that it converges reliably to a physically meaningful density matrix in practice. Contradicting these practical observations, we will show that convergence to a true solution is not guaranteed in general by constructing examples for spurious fixed points. To deal with this newly found problem, we then provide a criterion based on first order optimality conditions to check if the result of the algorithm is indeed the desired solution. Furthermore, we generalize the algorithm and show that it is equivalent to factorized gradient descent.
\\

\noindent

\noindent%
\end{abstract}

\section{Introduction}\label{sec:introduction}
The goal of quantum tomography is the determination of the state of a quantum system from a sequence of different measurements on copies of the state. It was theoretically proposed in \cite{Vogel:89} and realized experimentally in \cite{Smithey:93}. Since then it has become a standard tool in quantum optics \cite{Leonhardt:97} and was also recently applied to determine the quantum state of free-electrons \cite{Priebe:17}. Besides the measurement procedure, a numerical reconstruction algorithm is needed to compute the density matrix from the measured probabilities. Different methods have been proposed for this, like filtered back projection \cite{Vogel:89}, pattern functions \cite{Leonhardt:96}, positive semi-definite programming \cite{Priebe:17,Strandberg:22} or neural networks \cite{Torlai:18}. Another method that has stood the test of time and is popular with experimentalists \cite{Ourjoumtsev:06,Neergaard:06,Sychev:17,Hacker:19} is maximum likelihood iteration. It was proposed in \cite{Hradil:04} and later popularized and developed further in \cite{Lvovsky:04,Rehacek:07}. The main appeals of the method are its versatility, simplicity and the fact that in practice it converges reliably to a physically meaningful density matrix. 

Although some attempts have been made to prove the convergence of the algorithm \cite{Rehacek:07,Goncalves:14} these results remain incomplete. In fact, we will show that convergence to a true solution cannot be guaranteed in general. To the best of our knowledge this problem has not been observed in practice yet. However, in an experimental setting where other errors are present it could be difficult to distinguish a reasonable but suboptimal result from the true solution.

\subsection{Setting and notation}
For the treatment with convex analysis we use a slightly different notation than the one usually used in the physics context. While in principle our results could  be generalized to an infinite dimensional setting, we only consider finite dimensional problems with a fixed dimension \(N\). First, for a positive real number \(c\) we define the set
\begin{align*}
    \BopIpsdtr_{c}=\{\rho\in\Herm(N):\rho \text{ pos. semi-definite},\,\tr(\rho)=c\}
\end{align*}
where \(\Herm(N)\) is the space of hermitian matrices of size \(N\times N\).
The set of density matrices is then given by \(\BopIpsdtr_{1}\).
We represent a set of \(M\) tomography measurements by a single linear operator \(T:\Herm(N)\rightarrow \setR^{M\times N}\).  Furthermore, we assume that the operator maps valid density matrices to non-negative probabilities \(T\BopIpsdtr_{1}\subset [0,\infty)^{M\times N}\). We denote the measured probabilities by \(\yobs\) and assume that \(\yobs\in[0,\infty)^{M\times N}\). In practice, each row of measurements sums to one, but this is not needed for our analysis. Often the measurements are corrupted by noise such that \(\yobs\notin T\BopIpsdtr_{1}\). We therefore do not simply invert the operator and instead consider an optimization problem restricted to the set \(\BopIpsdtr_{1}\). An example for such an operator, that will be used later, is given by \(T_{2}:\Herm(2)\rightarrow\setR^{3\times 2}\) defined by
\begin{align*}
    (T_{2}\rho)_{m,n}=\tr(\rho\tau_{m,n})
\end{align*}
where the matrices \(\tau_{m,n}\) correspond to the outer products of the two different eigenvectors of the three Pauli matrices and are given by
\begin{align*}
    \tau_{1,1}&=\left(\begin{matrix}
        1 &0 \\
        0&0
    \end{matrix}\right)\quad
    &\tau_{2,1}&=\frac{1}{2}\left(\begin{matrix}
        1 &1 \\
        1&1
    \end{matrix}\right)\quad
    &\tau_{3,1}&=\frac{1}{2}\left(\begin{matrix}
        1 &-i \\
        i&1
    \end{matrix}\right)\\
    \tau_{1,2}&=\left(\begin{matrix}
        0 &0 \\
        0&1
    \end{matrix}\right)\quad
    &\tau_{2,2}&=\frac{1}{2}\left(\begin{matrix}
        1 &-1 \\
        -1&1
    \end{matrix}\right)\quad
    &\tau_{3,2}&=\frac{1}{2}\left(\begin{matrix}
        1 &i \\
        -i&1
    \end{matrix}\right).
\end{align*}
This operator then represents the standard six state tomography \cite{Wootters:89}.

As the name states, the approach of the maximum likelihood algorithm is to maximize the likelihood or alternatively minimize the negative log-likelihood \cite{Lvovsky:04}. In our notation this leads to

\begin{align}\label{prob:disc_kl}
    \underset{\rho\in\BopIpsdtr_{1}}{\argmin}-\ln\calL(\yobs,T\rho)=\underset{\rho\in\BopIpsdtr_{1}}{\argmin}-\sum_{m=1}^{M}\sum_{n=1}^{N}\yobs_{m,n}\ln (T\rho)_{m,n}.
\end{align}

The main challenge for the numerical solution of this problem lies in handling the restriction of the solution space onto the set of positive semi-definite matrices. For this, we define the indicator functionals \(\chi_{c}:\Herm(N)\rightarrow [0,\infty]\)
\begin{align*}
    \chi_{c}(\rho):=\begin{cases}
        0 &\quad \rho\in\BopIpsdtr_{c}\\
        \infty &\quad \text{else.}
    \end{cases}
\end{align*}
We will need the concept of subdifferentials which are a set valued generalization of gradients. For a convex functional \(\calF:X\rightarrow (-\infty,\infty]\) on a locally convex vector space \(X\) they are defined by
\begin{align*}
    \partial\calF(x)=\{x^{*}\in X^{*}:\calF(y)\geq \calF(x)-\langle x^{*},y-x\rangle\,\forall y\in X\}.
\end{align*}
For a functional \(\calF\) defined on \(\Herm(N)\) this becomes
\begin{align*}
    \partial\calF(\rho)=\{\delta\in \Herm(N):\calF(\sigma)\geq \calF(\rho)-\tr(\delta^{*}(\sigma-\rho))\,\forall \sigma\in \Herm(N)\}.
\end{align*}
Finally, to study fixed points we define for \(\rho\in\Herm(N)\) the set 
\begin{align*}
    \calM(\rho):=\{\sigma\in\Herm(N):\rho\sigma=\lambda\rho,\lambda\in\setR\}.
\end{align*}
of matrices which just scale \(\rho\) when multiplied. Note that because we only consider hermitian matrices, it does not matter if one considers left or right multiplications here.
\subsection{The standard maximum likelihood iteration}
With our notation the standard fixed point iteration from \cite{Lvovsky:04} can be written as
\begin{align}\label{al:iter_max_lik_std}
    \rho_{k+1}=\frac{R(\rho_{k})\rho_{k}R(\rho_{k})}{\tr(R(\rho_{k})\rho_{k}R(\rho_{k}))}\quad\text{with }R(\rho):=T^{*}\left(\frac{\yobs}{T\rho}\right)
\end{align}
where the division in the definition of \(R\) has to be carried out element-wise. In this form, it can be seen as a variant of the Richardson-Lucy algorithm from image processing, see e.g. \cite{Natterer:01}, which is modified so that positive semi-definiteness is preserved instead of positivity. The iteration relies on the observation that for operators from quantum tomography usually \(T^{*}\one=MI\) holds. So for \(\rhofix\) such that \(T\rhofix=\yobs\) one then gets \(R(\rhofix)=T^{*}\one=MI\) and it is a fixed point. We stress that the opposite direction does not hold if \(\rhofix\) does not have full rank. Additionally, this motivation assumes that \(T\rhofix=\yobs\) has a solution in the set of density matrices. Surprisingly, we will later see that the strong assumption \(T^{*}\one=MI\) is actually not really needed for the algorithm to work in practice.

In \cite{Rehacek:07} it was discovered, that this original iteration can fail to converge and an adapted version was proposed where \[R_{\epsilon}(\rho):=\epsilon T^{*}\left(\frac{\yobs}{T\rho}\right)+I\] with a changing step size \(\epsilon\). Assuming a suitable choice of step size, convergence was then proven but it was implicitly assumed that the solution computed by the algorithm has full rank. This is a problem since many physically meaningful states are represented by density matrices with low rank. Furthermore, one of the main advantages of the iteration is the incorporation of the positive semi-definiteness constraint. However, if this constraint is actually needed, meaning \(T^{-1}\yobs\notin \BopIpsdtr_{1}\), the solution to the constrained problem almost always does not have full rank. This is caused by the fact that the trace constraint is usually obeyed by the data due to normalization of the probabilities and only the positive semi-definiteness is not fulfilled which leads to a minimum of the constrained problem that lies on the boundary of the cone of positive semi-definite matrices. Finally, because the full rank assumption is made for the computed solution, even if the true solution to the problem does have full rank, the algorithm could get stuck at a fixed point where the iterate does not have full rank. Only cases in which the computed and the true solution have full rank are covered by the previous theory, but then a simple inversion would already produce reasonable results, and using the algorithm would not be necessary. In the following sections, we will provide a counterexample to the convergence and conduct a more thorough analysis of the iteration. Also, we will show that it is a variant of a factorized gradient descent algorithm as observed for standard maximum likelihood iteration in \cite{Shang:17}.
\section{Examples for spurious fixed points}\label{sec:example}
We now construct an example which shows that there are fixed points of the algorithm that are not the desired solutions to the problem. First, we provide a general strategy and then construct an explicit example for the operator \(T_{2}\). For given data \(\yobs\), finding a spurious fixed point \(\rhofix\) analytically is very challenging because of the nonlinear dependence of \(R(\rho)\rho\) on \(\rho\). We therefore use a different approach and start with a matrix \(\rhofix\) which we want to be the fixed point but leave the data to be undetermined. To avoid complications, we additionally assume that \(T^{*}\) is invertible and that \(T\rhofix\) is non-zero everywhere. Then for all \(\sigma\in\calM(\rhofix)\) the data defined by \(y(\sigma):=(T\rhofix)\odot(T^{*})^{-1}\sigma\) satisfies
\begin{align*}
    R(\rhofix,y(\sigma))=T^{*}\left(\frac{y(\sigma)}{T\rhofix}\right)=T^{*}\left(\frac{(T\rhofix)\odot(T^{*})^{-1}\sigma}{T\rhofix}\right)=\sigma\in\calM(\rhofix).
\end{align*}
Note that in the previously analyzed case, one has \(\sigma=\lambda I\) because \(\rhofix\) is assumed to have full rank and then only elements of the form \(\lambda I\) are in \(\calM(\rhofix)\). This then implies \(y(\lambda I)=(T\rhofix)\odot(T^{*})^{-1}(\lambda I)=(T\rhofix)\odot\frac{\lambda}{M}\one=\frac{\lambda}{M}T\rhofix\). So we do not get spurious fixed points. However, if this is not the case spurious fixed points can appear. If one would not impose restriction on the entries in the data, we could simply compute them with the previous formula. But to make sure that they are not the true solution to our problem, but correspond to a potentially different density matrix, we have to enforce \(T^{-1}y(\sigma)\in\BopIpsdtr_{1}\).

We will now do this for a \(2\times 2\)-matrix and the operator \(T_{2}\) from \Cref{sec:introduction}. We choose 
\begin{align*}
    \rhofix=\frac{1}{3}\left(\begin{matrix}
        1&1-i\\
        1+i&2
    \end{matrix}\right)\text{ and }\sigma=3I+3t\left(\begin{matrix}
        2&-1+i\\
        -1-i&1
    \end{matrix}\right).
\end{align*}
A simple calculation shows that indeed \(\rhofix\sigma=3\rhofix\) so \(\sigma\in\calM(\rhofix)\). Furthermore, we get
\begin{align*}
    T_{2}\rhofix=\frac{1}{6}\left(\begin{matrix}
        2&4\\
        5&1\\
        5&1
    \end{matrix}\right)\text{ and }y(\sigma)=\frac{1}{6}\left(\begin{matrix}
        2+4t&4-4t\\
        5-5t&1+5t\\
        5-5t&1+5t
    \end{matrix}\right).
\end{align*}
We can then compute
\begin{align*}
    T_{2}^{-1}y(\sigma)=\frac{1}{6}\left(\begin{matrix}
        2+4t&(2-5t)(1-i)\\
        (2-5t)(1+i)&4-4t
        \end{matrix}\right)
\end{align*}
which is positive semi-definite if \(t\in[0,\frac{8}{11}]\). We only have \(T_{2}^{-1}y(\sigma)=\rhofix\) for \(t=0\) which means that for all other \(t\) we have found data \(y(\sigma)\) such that \(\rhofix\) is a fixed point, but not the solution to problem \ref{prob:disc_kl}. Note that this is not caused by any missing data and it is an intrinsic problem with the fixed point iteration. In addition to this, for \(t\in(0,\frac{8}{11})\) the true solution has full rank, which shows that the algorithm can get stuck even if the data corresponds to a valid full rank density matrix. These problems also persist if we do the more complicated iteration with a variable step size as this does not change the fixed point set.

\section{The Gradient Multiplication Algorithm}\label{sec:algorithm_general}
We now generalize the algorithm and provide a theorem based on first order optimality conditions that helps to distinguish spurious fixed points from actual solutions. For the generalization, we consider problems of the form 
\begin{align}\label{prob:general}
    \underset{\rho\in\BopIpsdtr_{c}}{\argmin}\,\calF(\rho)
\end{align}
where \(\calF:\BopIpsdtr_{c}\rightarrow \setR\) is a convex and continuously differentiable functional. This problem structure is present for example when the negative log-likelihood or the \(L^{2}\)-norm is used to assess the fit to the data or when additional convex penalty terms are added for regularization purposes as it is done in the SQUIRRELS algorithm \cite{Priebe:17}. The generalization to other fixed trace values also allows the algorithm to be applied in situations outside of quantum tomography e.g. for the determination of covariance matrices.

The generalized algorithm, which we call Gradient Multiplication algorithm because of its structure, is again given by a simple fixed point iteration
\begin{align}\label{al:iter_RI_general}
    \rho_{k+1}=c\frac{(I-\epsilon\nabla\calF(\rho_{k}))\rho_{k}(I-\epsilon\nabla\calF(\rho_{k}))}{\tr((I-\epsilon\nabla\calF(\rho_{k}))\rho_{k}(I-\epsilon\nabla\calF(\rho_{k})))}\quad\text{ with }\rho_{0}\in\BopIpsdtr_{c}.
\end{align}
For the setting where \(\calF(\rho)=-\ln\calL(\yobs,T\rho)\) this is exactly the same algorithm as the maximum likelihood iteration with variable step size from \cite{Rehacek:07}.

\subsection{Convergence properties}
The main difference in our analysis of the algorithm is that the assumption \(T^{*}\one=MI\) is no longer needed. Instead the algorithm generally relies on the fact that the subgradient of the indicator functional \(\chi_{c}\) and \(\calM\) are very similar and in fact the same for \(\rho\) with full rank. To see this, we first characterize the set \(\calM\).
\begin{lemma}
    For \(\rho\in\BopIpsdtr_{c}\)
    \begin{align*}
        \calM(\rho)=\{\lambda I-Q:\lambda \in \setR, Q\in\Herm(N), \ran Q\subseteq\ker \rho\} .
    \end{align*}
\end{lemma}
\begin{proof}
    That all matrices in the set just scale \(\rho\) follows from a simple computation. For the opposite direction assume that \(\sigma\in\calM(\rho)\) such that \(\rho\sigma=\lambda\rho\) and let \(P\) be the projection onto the kernel of \(\rho\). Then \(\rho+P\) is invertible and \((\rho+P)^{-1}\rho=I-P\). This leads to
    \begin{align*}
        (\rho+P)^{-1}\rho\sigma&=(\rho+P)^{-1}\lambda\rho\\
        (I-P)\sigma&=\lambda(I-P)\\
        \sigma&=\lambda I-P(\lambda I+\sigma)
    \end{align*}
    and \(\ran P(\lambda I+\sigma)\subseteq\ker \rho\) which finishes the proof.
\end{proof}
For the subgradient of \(\chi\) we get the following set.
\begin{lemma}
    For \(\rho\in\BopIpsdtr_{c}\)
    \begin{align*}
        \partial\chi_{c}(\rho)=\{\lambda I-Q:\lambda \in \setR, Q\in\Herm(N),\ran Q\subseteq\ker \rho,Q \text{ p.s.d.}\} .
    \end{align*}
\end{lemma}
\begin{proof}
    First, we show that every matrix of this form is in the subgradient. For this we take \(\sigma\in\BopIpsdtr_{c}\) and compute
    \begin{align*}
    \chi_{c}(\rho)+\tr((\lambda I-Q)(\sigma-\rho))&=0+\lambda(\tr(\sigma)-\tr(\rho))-\tr(Q\sigma)+\tr(Q\rho)\\
    &=-\tr(Q\sigma)=-\tr(Q^{\frac{1}{2}}\sigma^{\frac{1}{2}}\sigma^{\frac{1}{2}}Q^{\frac{1}{2}})\leq 0=\chi_{c}(\sigma).
\end{align*}
For the other direction assume that \(M\in\partial\chi_{c}(\rho)\). This then implies \(\tr(M\sigma)\leq\tr(M\rho)\) for all \(\sigma\in\BopIpsdtr\). Because \(M\in\Herm(N)\) we can diagonalize it and write \(M=:VDV^{*}\) which implies
\[\tr(D\sigma)\leq\tr(VDV^{*}\rho)=\tr(DV^{*}\rho V)=:\tr(D\tilde{\rho})\quad\forall\sigma\in\BopIpsdtr_{c}.\]
Now for \(j,k\) such that \(\tilde{\rho}_{j,j},\tilde{\rho}_{k,k}\neq 0\) we have
\begin{align*}
    \tr(D\tilde{\rho})=\sum_{i\neq j,k}D_{i,i}\tilde{\rho}_{i,i}+D_{j,j}\tilde{\rho}_{j,j}+D_{k,k}\tilde{\rho}_{k,k}<\begin{cases}
        \sum_{i\neq j,k}D_{i,i}\tilde{\rho}_{i,i}+D_{j,j}(\tilde{\rho}_{j,j}+\tilde{\rho}_{k,k})\quad \text{for } D_{j,j}>D_{k,k}\\
        \sum_{i\neq j,k}D_{i,i}\tilde{\rho}_{i,i}+D_{k,k}(\tilde{\rho}_{j,j}+\tilde{\rho}_{k,k})\quad \text{for } D_{k,k}>D_{j,j}
    \end{cases}
\end{align*}
and this implies that \(D_{j,j}\) is constant for all \(i\in\calI:=\{j:\tilde{\rho}_{j,j}\neq 0\}\).  We can then write \(D=:\lambda I-D_{q}\) with
\begin{align*}
    (D_{q})_{j,j}=\begin{cases}
        0\quad&\text{for }{j\in\calI}\\
        q_{j}\in\setR\quad&\text{for }{j\notin\calI}.
    \end{cases}
\end{align*}
Then by taking \(\sigma\) to be the matrix \(c\cdot e_{j}e_{j}^{*}\) for \(j\notin\calI\) we get
\begin{align*}
    \tr(D(c\cdot e_{j}e_{j}^{*}))\leq\tr(D\tilde{\rho})\Leftrightarrow
    \lambda c\tr(e_{j}e_{j}^{*})-q_{j}\leq \lambda\tr(\tilde{\rho})\Leftrightarrow q_{j}\geq 0
\end{align*}
So \(M=\lambda I-VD_{q}V^{*}\) and \(VD_{q}V^{*}\) is positive semi-definite. It remains to show that its range is in the kernel of \(\rho\). By positive semi-definiteness if \(\tilde{\rho}_{j,j}=0\) so is \(\tilde{\rho}_{k,j}\) which means \(0=\tilde{\rho}e_{j}=V^{*}\rho Ve_{j}\) and \(Ve_{j}\in\ker \rho\) for \(j\notin\calI\). Now because \(q_{j}=0\) for \(j\in\calI\) we get \(\ran VD_{q}V^{*}\subseteq \ker \rho\).
\end{proof}
We now reformulate problem (\ref{prob:general}) and by the convexity and the optimality conditions get
\begin{align*}
    \rho^{\dagger}&\in\underset{\rho\in\Herm(N)}{\argmin}\calF(\rho)+\chi_{c}(\rho)\\
    \Leftrightarrow\quad 0&\in \partial\calF(\rho^{\dagger})+\partial\chi_{c}(\rho^{\dagger})\\
    \Leftrightarrow\quad -\nabla\calF(\rho^{\dagger})&\in \partial\chi_{c}(\rho^{\dagger}).
\end{align*}
A first relationship between the iteration and the solution to problem \ref{prob:general} is given by the following lemma.
\begin{lemma}
    Let \(\rho^{\dagger}\) be a solution to problem (\ref{prob:general}), then it is a fixed point of iteration (\ref{al:iter_RI_general}) for all \(\epsilon\in\setR\setminus\{\mu\}\) for at most a single value \(\mu\), which is given by \(c\left(\tr(\nabla\calF(\rho^{\dagger})\rho^{\dagger})\right)^{-1}\).
\end{lemma}
\begin{proof}
    For such a solution \(\rho^{\dagger}\), we know that \(-\nabla\calF(\rho^{\dagger})\in\partial\chi_{c}(\rho^{\dagger})\) which is a subset of \(\calM(\rho^{\dagger})\). Therefore there exists some \(\lambda \in \setR\) such that \(-\nabla\calF(\rho^{\dagger})\rho^{\dagger}=\lambda\rho^{\dagger}\) and this implies
    \begin{align*}
        (I-\epsilon\nabla\calF(\rho^{\dagger}))\rho^{\dagger}(I-\epsilon\nabla\calF(\rho^{\dagger}))=\rho^{\dagger}+2\lambda\epsilon\rho^{\dagger}+\epsilon^{2}\lambda^{2}\rho^{\dagger}=(1+\epsilon\lambda)^{2}\rho^{\dagger}.
    \end{align*}
    So \(\rho^{\dagger}\) is a fixed point for \(\epsilon\neq-\lambda^{-1}\). To calculate \(\lambda\) we use \(\lambda=c^{-1}\tr(\lambda\rho^{\dagger})=c^{-1}\tr(-\nabla\calF(\rho^{\dagger})\rho^{\dagger})\).
\end{proof}
We also show that for \(\epsilon\) chosen small enough the iteration decreases the objective function.
\begin{lemma}
    Let \(\rho_{k}\in\BopIpsdtr_{c}\). If \(\rho_{k}\) is not a fixed point of (\ref{al:iter_RI_general}), then for all \(\epsilon\) small enough
    \[\calF(\rho_{k+1})<\calF(\rho_{k}).\]
\end{lemma}
\begin{proof}
    We follow and generalize the calculations in \cite{Rehacek:07}. As it is done there we focus on linear terms in \(\epsilon\) and compute
    \begin{align*}
        \rho_{k+1}&=c\frac{\rho_{k}-\epsilon(\nabla\calF(\rho_{k})\rho_{k}+\rho_{k}\nabla\calF(\rho_{k}))+\epsilon^{2}\nabla\calF(\rho_{k})\rho_{k}\nabla\calF(\rho_{k})}{\tr\left(\rho_{k}-\epsilon(\nabla\calF(\rho_{k})\rho_{k}+\rho_{k}\nabla\calF(\rho_{k}))+\epsilon^{2}\nabla\calF(\rho_{k})\rho_{k}\nabla\calF(\rho_{k})\right)}\\
        &=c\frac{\rho_{k}}{\tr(\rho_{k})}+c\epsilon\frac{-(\nabla\calF(\rho_{k})\rho_{k}+\rho_{k}\nabla\calF(\rho_{k}))\tr(\rho_{k})+2\tr(\nabla\calF(\rho_{k})\rho_{k})\rho_{k}}{\tr(\rho_{k})^{2}}+\calO(\epsilon^{2})\\
        &=\rho_{k}+\epsilon\left[-(\nabla\calF(\rho_{k})\rho_{k}+\rho_{k}\nabla\calF(\rho_{k}))+\frac{2}{c}\tr(\nabla\calF(\rho_{k})\rho_{k})\rho_{k}\right]+\calO(\epsilon^{2}).
    \end{align*}
    Then for \(\sigma\) such that \(\rho_{k}+\sigma\) is positive semidefinite, we have that \(\rho_{k}+\epsilon\sigma\) is positive semidefinite for all \(\epsilon\leq1\) and
    \begin{align*}
        \calF(\rho_{k}+\epsilon\sigma)&=\calF(\rho_{k})+\epsilon\tr(\nabla\calF(\rho_{k})\sigma)+\calO(\epsilon^{2})
    \end{align*}
    Combining both equalities we get
    \begin{align*}
        \calF(\rho_{k})&-\calF(\rho_{k+1})\\
        &=-\epsilon\tr\left(\nabla\calF(\rho_{k})\left[-(\nabla\calF(\rho_{k})\rho_{k}+\rho_{k}\nabla\calF(\rho_{k}))+\frac{2}{c}\tr(\nabla\calF(\rho_{k})\rho_{k})\rho_{k}\right]\right)+\calO(\epsilon^{2})\\        
        &=\frac{2}{c}\epsilon(c\tr(\nabla\calF(\rho_{k})^{2}\rho_{k})-\tr(\nabla\calF(\rho_{k})\rho_{k})^{2})+\calO(\epsilon^{2}).
    \end{align*}
    As a last step we use the Cauchy-Schwarz inequality to get
    \[c\tr(\nabla\calF(\rho_{k})^{2}\rho_{k})=\tr(\nabla\calF(\rho_{k})^{2}\rho_{k})\tr(\rho_{k})\geq \tr(\nabla\calF(\rho_{k})\rho_{k})^{2}\]
    where equality holds only if \(\nabla\calF(\rho_{k})\rho_{k}^{\frac{1}{2}}\) is a scalar multiple of \(\rho_{k}^{\frac{1}{2}}\) which is equivalent to \(-\nabla\calF(\rho_{k})\in\calM(\rho_{k})\). So if this is not the case we can choose \(\epsilon\) small enough such that the linear part dominates and the claimed inequality holds strictly.
\end{proof}
Note that the proof of this lemma does not rely on the convexity of \(\calF\). So the decay of the functional during the iteration can also be shown for non-convex functionals.

Using the continuity of \(\nabla\calF\), we can now use the same argument as in \cite{Rehacek:07} to show that choosing \(\epsilon\) at each step such that it maximizes the reduction of \(\calF\) guarantees convergence to a fixed point. For such a fixed point we can then check if it is the global minimum.
\begin{theorem}\label{the:fix_opt_cond_general}
    Assuming that \(\epsilon\) is chosen at each step such that the reduction of \(\calF\) is maximized, the iteration (\ref{al:iter_RI_general}) converges to a fixed point \(\rho^{\dagger}\). This is a solution to problem (\ref{prob:general}) if and only if it satisfies
    \begin{align*}
        \nabla\calF(\rho^{\dagger})-\tr(\nabla\calF(\rho^{\dagger})\rho^{\dagger})I \text{ is positive semi-definite.}
    \end{align*}
\end{theorem}
\begin{proof}
    The convergence follows from the previous lemma together with the argument from the appendix of \cite{Rehacek:07}. For the validity condition assume that \(\rho^{\dagger}\) is a fixed point. Then \(-\nabla\calF(\rho^{\dagger})\in\calM(\rho^{\dagger})\) and \(-\nabla\calF(\rho^{\dagger})=\lambda I-Q\) with \(\ran Q\subset \ker \rho^{\dagger}\). This means
    \begin{align*}
        \tr(-\nabla\calF(\rho^{\dagger})\rho^{\dagger})I+\nabla\calF(\rho^{\dagger})=\tr((\lambda I-Q)\rho^{\dagger})I-\lambda I+Q=\lambda I-\lambda I+Q=Q.
    \end{align*}
    Now if \(Q\) is positive semidefinite, \(-\nabla\calF(\rho^{\dagger})\in\partial\chi_{c}(\rho^{\dagger})\) and it is therefore a solution to problem \ref{prob:general}. In the same way if \(Q\) is not positive semidefinite \(-\nabla\calF(\rho^{\dagger})\notin\partial\chi_{c}(\rho^{\dagger})\) and it is not a solution to the problem.
\end{proof}
In practice it has been recommended  so far to fix a step size at the start and keep track of the decay of the target function \(\calF\) \cite{Rehacek:07}. The step size is then just decreased if the current step size does not lead to a decay of the target function. This also works well in our generalized setup and often the step size is fixed for almost the whole algorithm. 

Based on our results one should additionally check for the fixed point if the condition for \Cref{the:fix_opt_cond_general} is satisfied. There are some numerical challenges with this criterion, because \(\rank Q\leq \dim\ker\rho^{\dagger}\) leads to many eigenvalues that are zero and may become slightly negative due to numerical inaccuracies. If it is feasible we therefore recommend to consider the restriction of \(Q\) onto \(\ker\rho^{\dagger}\) instead and check that it is positive semidefinite to get rid of many of the zero eigenvalues. This has the downside that if  \(-\nabla\calF(\rho^{\dagger})=\lambda I-Q\) also only holds approximately one may accidentally ignore some slightly negative eigenvalues by just considering this restriction.
If the criterion fails one could always restart the algorithm with a different step size or starting point. A more efficient alternative is to instead continue with a different method like projected gradient descent and use the obtained result as a good starting point.

\Cref{the:fix_opt_cond_general} also provides us with a new method to numerically construct spurious fixed points if the rank of the true solution is known. 
\begin{corollary}\label{cor:gen_spurious}
    If the true solution \(\rho^{\dagger}\) to problem \ref{prob:general}  has rank greater than one and iteration \ref{al:iter_RI_general} is initialized with a matrix \(\rho_{0}\in\BopIpsdtr_{c}\) such that \(1\leq\rank(\rho_{0})<\rank(\rho^{\dagger})\), then the iteration converges to a spurious fixed point.
\end{corollary}
\begin{proof}
    From \Cref{the:fix_opt_cond_general} we get that the iteration converges as long as \(\rho_{0}\in\BopIpsdtr_{c}\). So we can define the limit of the iteration to be \(\rho_{\text{fix}}\). Now because the algorithm only relies on scalar divisions and matrix multiplications, we have \(\rank(\rho_{k+1})\leq\rank(\rho_{k})\). By our assumption we then get \(\rank(\rho_{\text{fix}})\leq\rank(\rho_{0})<\rank(\rho^{\dagger})\) which implies \(\rho_{\text{fix}}\neq\rho^{\dagger}\).
\end{proof}
We will use this fact in \Cref{sec:experiments} to obtain spurious fixed points for larger operators.

\subsection{Equivalence to factored gradient descent}
So far we have analyzed the algorithm in the context of convex optimization. Now we will show, that it can also be viewed as a gradient descent algorithm for the non-convex problem

\begin{align*}
    \underset{X\in\setC^{N\times N},\|X\|_{F}=\sqrt{c}}{\argmin}\calF(XX^{*}).
\end{align*}
A gradient descent method for this problem was analyzed in \cite{Kyrillidis:17} and is given by
\begin{align}\label{al:iter_grad_descent}
    X_{k+\frac{1}{2}}=X_{k}-\epsilon \nabla\calF(X_{k}X_{k}^{*})X_{k}
    ,\quad X_{k+1}=\frac{\sqrt{c}X_{k+\frac{1}{2}}}{\|X_{k+\frac{1}{2}}\|_{F}}.
\end{align}
This turns out to be exactly the same iteration as in the Gradient Multiplication algorithm and therefore also in the maximum likelihood algorithm. For the standard maximum likelihood algorithm, this has been observed already in \cite{Shang:17} and in the context of coherence retrieval this was proposed in \cite{Zhang:13}. There the possibility of spurious fixed points was discussed as well. We generalize these results.
\begin{theorem}\label{the:equiv}
    If the algorithms (\ref{al:iter_RI_general}) and (\ref{al:iter_grad_descent}) are initialized such that \(\rho_{0}=X_{0}X_{0}^{*}\) and if the same step size is chosen in each step for both algorithms, then \(\rho_{k}=X_{k}X_{k}^{*}\) holds for all iterates.
\end{theorem}
\begin{proof}
    We prove this statement by induction. The base case is given by assumption. For the induction step we assume that \(\rho_{k}=X_{k}X_{k}^{*}\) and then calculate
    \begin{align*}
    \rho_{k+1}&=c\frac{(I-\epsilon\nabla\calF(\rho_{k}))\rho_{k}(I-\epsilon\nabla\calF(\rho_{k}))}{\tr((I-\epsilon\nabla\calF(\rho_{k}))\rho_{k}(I-\epsilon\nabla\calF(\rho_{k})))}=\frac{(I-\epsilon\nabla\calF(X_{k}X_{k}^{*}))X_{k}X_{k}^{*}(I-\epsilon\nabla\calF(X_{k}X_{k}^{*}))}{\tr((I-\epsilon\nabla\calF(X_{k}X_{k}^{*}))\rho_{k}(I-\epsilon\nabla\calF(X_{k}X_{k}^{*})))}\\
    &=c\frac{\left[(I-\epsilon\nabla\calF(X_{k}X_{k}^{*}))X_{k}\right]\left[(I-\epsilon\nabla\calF(X_{k}X_{k}^{*}))X_{k}\right]^{*}}{\tr(\left[(I-\epsilon\nabla\calF(X_{k}X_{k}^{*}))X_{k}\right]\left[(I-\epsilon\nabla\calF(X_{k}X_{k}^{*}))X_{k}\right]^{*})}\\
    &=\frac{\sqrt{c}(I-\epsilon\nabla\calF(X_{k}X_{k}^{*}))X_{k}}{\|(I-\epsilon\nabla\calF(X_{k}X_{k}^{*}))X_{k}\|_{F}}\left(\frac{\sqrt{c}(I-\epsilon\nabla\calF(X_{k}X_{k}^{*}))X_{k}}{\|(I-\epsilon\nabla\calF(X_{k}X_{k}^{*}))X_{k}\|_{F}}\right)^{*}=X_{k+1}X_{k+1}^{*}.
\end{align*}
\end{proof}
This means that instead of the Gradient Multiplication algorithm one could always do factorized gradient descent. This gradient descent algorithm has the additional advantage, that it can be performed on \(X\in\setC^{N\times r}\) instead to produce density matrices with rank bounded by \(r\).

This factorization procedure also applies to the original maximum likelihood algorithm (\ref{al:iter_max_lik_std}) without step size which can then be performed by computing \cite{Shang:17}
\begin{align}\label{al:iter_max_lik_factorized}
    X_{k+1}=\frac{R(X_{k}X_{k}^{*})X_{k}}{\|R(X_{k}X_{k}^{*})X_{k}\|_{F}}\quad\text{with }R(X_{k}X_{k}^{*}):=T^{*}\left(\frac{\yobs}{T(X_{k}X_{k}^{*})}\right).
\end{align}
Because \(R(X_{k}X_{k}^{*})X_{k}\) is half the negative gradient of \(-\ln\calL(\yobs,T(XX^{*}))\), we can see that the original algorithm works by replacing an iterate with its scaled negative gradient. This provides us with a geometric interpretation of the algorithm. With this parametrization, all possible solutions lie on a unit ball. Due to the monotonicity of the negative log-likelihood function, the negative gradients all point outwards as depicted schematically in \Cref{fig:geo_max_lik}. At each vector corresponding to the true solution, they point exactly in the same direction as this vector and otherwise they get closer and closer to pointing in its direction. One can therefore hope to get to the true solution by replacing each vector with its negative gradient. The appearance of cycles as described in \cite{Rehacek:07} is than caused by vectors which are each others negative gradients.
\begin{figure}
    \centering
    \includegraphics[scale=0.4]{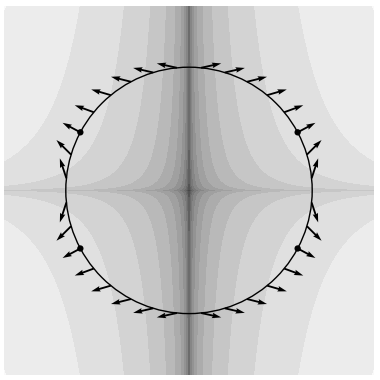}
    \caption{Schematic illustration of the geometry behind the (factorized) maximum likelihood iteration (\ref{al:iter_max_lik_factorized}). The values of the function \(-a^{2}\ln(x^{2})-b^{2}\ln(y^{2})\) for \(a\approx 0.88,b\approx 0.47\) are indicated by gray values, and the directions of the negative gradients for points on the unit circle are represented by arrows. The minima \((\pm a,\pm b)\) on the unit circle are marked with black dots.}
    \label{fig:geo_max_lik}
\end{figure}

\section{Numerical experiments}\label{sec:experiments}
To numerically validate our findings from the previous section, we apply them to the example of homodyne tomography. The code to reproduce all of them is provided in \cite{Oberender:25data}. In this setting the (binned) forward operator is given by 
\begin{align*}
   (T\rho)(\theta,k)=\int_{x_{k}}^{x_{k+1}}\sum_{m,n\in\setN}\rho_{m,n}e^{i(n-m)\theta}h_{m}(x)h_{n}(x)\dd x 
\end{align*}
with \(h_{m}\) being the Hermite function of order \(m\) \cite{Richter:00}. We consider the problem of reconstructing a \(10\times 10\) matrix from a measurement with 15 angles and 50 bins in the \(x\)-direction spaced equally on the interval \([-7,7]\). For data \(\yobs\) we consider the optimization problem
\begin{align}\label{prob:special}
    \underset{\rho\in\BopIpsdtr_{1}}{\argmin}\,\calS_{\yobs}(T\rho)
\end{align}
for \(\calS_{\yobs}=-\ln\calL(\yobs,\cdot)\) corresponding to maximum likelihood iteration and for \(\calS_{\yobs}=\frac{1}{2}\|\yobs-\cdot\|_{2}^{2}\) which models a standard least squares error. This is done to demonstrate that the Gradient Multiplication algorithm can be applied outside of the maximum likelihood context.

We generated a data set of 1000 random density matrices such that we have 100 for each rank. The exact data is then simulated by applying the forward operator \(T\). We also added Poisson noise by computing \(y^{\text{noisy}}\sim\pois(500\odot T\rho) \) and normalizing for each angle \(\theta\). This leads to a relative noise level of around 22 percent with a slight variation coming from the normalization.

We did the reconstruction for both different functionals, with and without noise with the Gradient multiplication algorithm and with factorized gradient descent until convergence or until a maximum of 20,000 iterations was reached. For exact data we compare the result to the exact solution in the trace norm. To do a similar comparison for the noisy version we compute a solution to problem (\ref{prob:special}) using Cvxpy \cite{Diamond:16} as presented in \cite{Strandberg:22}. 

\begin{figure}
\begin{subfigure}[b]{0.45\textwidth}
    \centering
    \includegraphics[scale=0.33]{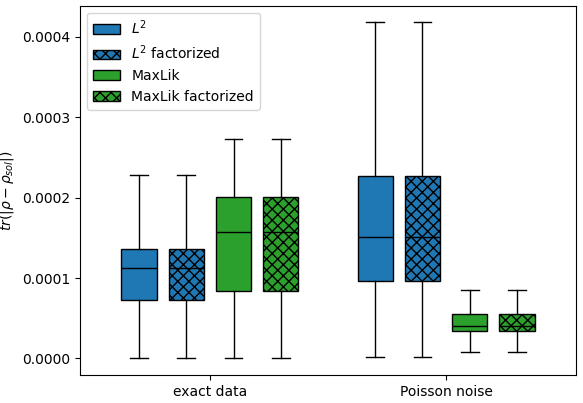}
    \caption{}
\end{subfigure}
\hspace{0.5cm}
\begin{subfigure}[b]{0.45\textwidth}
    \centering
    \includegraphics[scale=0.33]{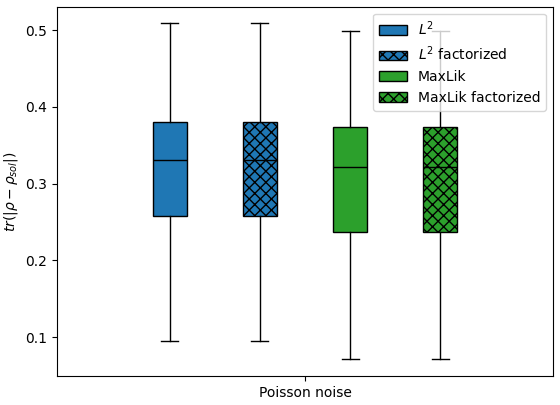}
    \caption{}
\end{subfigure}
\caption{(a) Difference in trace norm between the exact solution to problem (\ref{prob:special}) and the one computed with the respective method for data with and without noise. (b) Difference in trace norm between the solutions obtained with the algorithms applied to data with Poisson noise and the exact solution.}
\label{fig:comp_norms}
\end{figure}

The results are depicted in \Cref{fig:comp_norms}a. For the exact data both versions produce very similar results, with the \(L^{2}\)-norm working slightly better. For noisy data the result for the maximum likelihood iteration is about one order of magnitude smaller than the one with the \(L^{2}\)-norm. As expected, there is no difference between the results of the maximum likelihood iteration and the factorized gradient descent.

For completeness we also compared the solutions for Poisson noise to the exact solution (see \Cref{fig:comp_norms}b). We can see that the results for both versions are again very similar and even though the maximum likelihood approach is based on a more accurate noise model, the results are only slightly more accurate on average. This is evidence for our hypothesis, that the main effectiveness of the maximum likelihood approach so far is not caused by modeling the likelihood in a more accurate way but instead by incorporating the restriction to the set of density matrices into the reconstruction process.  

\begin{figure}
\begin{subfigure}[b]{0.45\textwidth}
    \centering
    \includegraphics[scale=0.33]{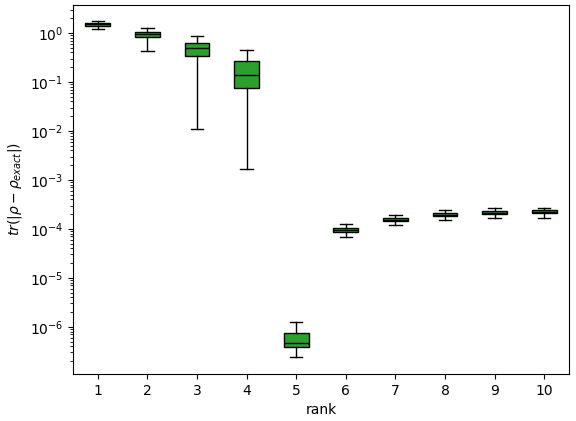}
    \caption{}
\end{subfigure}
\hspace{0.5cm}
\begin{subfigure}[b]{0.45\textwidth}
    \centering
    \includegraphics[scale=0.33]{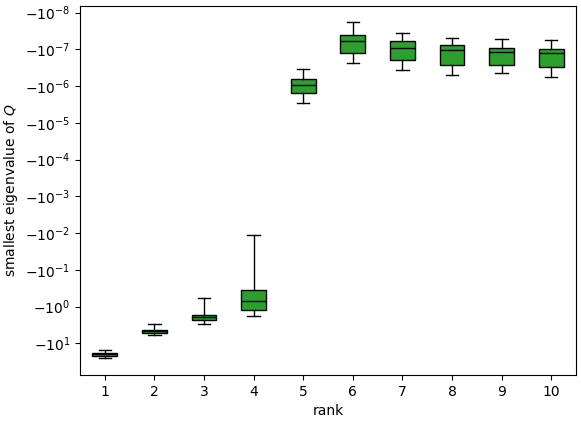}
    \caption{}
\end{subfigure}
\caption{(a) Difference in trace norm between the exact solution to problem (\ref{prob:special}) with rank 5 and the one computed with the maximum likelihood iteration depicted for different ranks of the starting points. (b) Smallest eigenvalues of \(Q=\nabla\calF(\rho)-\tr(\nabla\calF(\rho)\rho)I\), which should be positive semidefinite for the true solution by \Cref{the:fix_opt_cond_general}, for different starting ranks.}
\label{fig:comp_counter}
\end{figure}
We also used the result from \Cref{cor:gen_spurious} and the factorized gradient descent method to generate spurious fixed points. For this we considered as the set of true solutions the random density matrices with rank 5 and then did the reconstruction starting with a matrix with rank \(r\) for \(r\in\{1,...,10\}\). The results in \Cref{fig:comp_counter}a show, that the true solution is only (approximately) reached if \(r\geq 5\) as expected. For \(r<5\) we have reached spurious fixed points. We also tested the criterion from \Cref{the:fix_opt_cond_general}. The smallest eigenvalues of the matrices \(Q=\nabla\calF(\rho)-\tr(\nabla\calF(\rho)\rho)I\) depicted in \Cref{fig:comp_counter}b show again, that the fixed points reached for \(r<5\) do not correspond to the true solution. For \(r\geq 5\) the eigenvalues are still negative due to numerical inaccuracies but much closer to zero.

The numerical experiments confirm our results. Most importantly they show that spurious fixed points also exists in practically relevant setups and that maximum likelihood iteration and factorized gradient descent produce the same results.

\section{Conclusions}\label{sec:conclusions}
We analytically and numerically showed the existence of spurious fixed points in  the iterative maximum likelihood algorithm and found a criterion to distinguish them from actual solutions. We furthermore generalized the algorithm and showed that it is equivalent to gradient descent on a factorized matrix. This was also confirmed numerically, and we showed that other data fidelity functionals such as the \(L^{2}\)-norm lead to similar results. 

Apart from providing a detection method for spurious fixed points to experimentalists, these findings could be used to develop faster variants of the maximum likelihood algorithm and incorporate rank constraints.

\section*{Data Availability Statement}
The code associated with this article is available in `GRO.data', under the reference

\noindent https://doi.org/10.25625/GPKB4Q.
\bibliographystyle{abbrv}
\bibliography{references}

\end{document}